\documentclass[lettersize,journal]{IEEEtran}

\usepackage{amsmath,amssymb,amsfonts,amsthm}
\usepackage{algorithm}
\usepackage{algpseudocode}
\usepackage{array,tabularx,makecell,multirow}
\usepackage[caption=false,font=normalsize,labelfont=sf,textfont=sf]{subfig}
\usepackage{textcomp}
\usepackage{stfloats}
\usepackage{url}
\usepackage{verbatim}
\usepackage{graphicx}
\usepackage{graphics}
\usepackage{cite}
\usepackage{xcolor}
\usepackage[top=0.69in, bottom=1.02in, left=0.63in, right=0.63in]{geometry}

\theoremstyle{plain}
\newtheorem{theorem}{\textit{Theorem}}
\newtheorem{lemma}{Lemma}
\newcounter{assumption}

\theoremstyle{remark}

\begin{document}

	\title{Performance Analysis of Fluid Antenna–Assisted Over-the-Air Federated Learning Under Spatially Correlated Fading
}
	
	\author{IEEE Publication Technology,~\IEEEmembership{Staff,~IEEE,}
	}

	\author{
		\IEEEauthorblockN{Mohsen Ahmadzadeh, Saeid Pakravan, Wessam Ajib, Ming Zeng, Ghosheh Abed Hodtani, and Ji Wang}

	\thanks{M. Ahmadzadeh and G. Abed Hodtani are with the Department of Electric and Computer Engineering, Ferdowsi University, Mashhad, Iran. email: m.ahmadzadehbolghan@mail.um.ac.ir; hodtani@um.ac.ir.}

\thanks{S. Pakravan and W. Ajib are with the Department of Computer Sciences, University of Quebec in Montreal (UQAM), Montreal, QC, Canada. email: pakravan.saeid@uqam.ca; ajib.wessam@uqam.ca.}

\thanks{M. Zeng is with the Department of Electrical and Computer Engineering, Laval University, Quebec City, QC, CA. email: ming.zeng@gel.ulaval.ca.}

\thanks{J. Wang is with the Department of Electronics and Information Engineering, College of Physical Science and Technology, Central China Normal University, Wuhan 430079, China. email: jiwang@ccnu.edu.cn.}

	}

	\maketitle

		\begin{abstract}
Fluid antenna (FA) technology has recently emerged as an effective means of exploiting spatial diversity through position-domain reconfigurability. This paper investigates the integration of FA into over-the-air federated learning (OTA-FL) systems with the aim of improving aggregation reliability and user participation under realistic channel conditions. By dynamically selecting antenna positions, FA-equipped users can exploit additional spatial degrees of freedom to realize more favorable channel conditions, thereby increasing the probability of successful contribution to the OTA aggregation process in each communication round. We consider an uplink OTA-FL framework consisting of a single fixed-antenna access point and multiple FA-enabled users operating over spatially correlated fading channels. Unlike existing studies that primarily rely on optimization-based designs or numerical evaluations, we develop a tractable analytical framework that enables a rigorous performance characterization of FA-assisted OTA-FL. In particular, closed-form expressions are derived for the aggregation error outage probability and the expected number of participating users per round. Spatial channel correlation across FA ports is modeled using a copula-based approach, where the Clayton copula is adopted to capture lower-tail dependence relevant to worst-case fading conditions. Numerical results validate the analytical findings and demonstrate that FA-assisted OTA-FL significantly outperforms conventional fixed-antenna schemes in terms of aggregation reliability and participation efficiency, while providing insights under practical system considerations.

\end{abstract}


\section{Introduction}

The advent of sixth-generation (6G) wireless networks is expected to enable pervasive intelligence by tightly integrating communication, sensing, and computing functionalities. In this context, federated learning (FL) has emerged as a fundamental paradigm for distributed machine learning, particularly suited to data-driven applications such as the Internet of Things (IoT), autonomous transportation systems, and Industry~4.0 \cite{9829367,7777777}. Unlike centralized learning, FL allows a large number of edge devices to collaboratively train a shared global model while retaining raw data locally, thereby alleviating privacy concerns and reducing the risk of data leakage \cite{9141214, 9264742}. These features make FL especially attractive for 6G networks, where massive connectivity and stringent privacy requirements coexist.

Despite its conceptual appeal, the practical deployment of FL over wireless networks faces severe communication bottlenecks. In each training round, a potentially large number of devices must transmit high-dimensional model updates to a central parameter server or access point (AP). This process incurs substantial latency, consumes excessive spectral resources, and places a heavy burden on the energy-constrained edge devices. As the number of participating users increases, these challenges can limit scalability, convergence speed, and learning accuracy of FL systems, making communication-efficient aggregation mechanisms essential for large-scale deployment.

Over-the-air (OTA) computation has emerged as a promising solution to this problem, particularly in uplink edge-intelligence scenarios where a large number of devices communicate with a central AP over a shared wireless multiple-access channel under stringent latency and bandwidth constraints \cite{10538293, 10092857}. Such settings naturally arise in applications such as industrial IoT, smart sensing, and distributed monitoring, where many edge devices must frequently upload model updates for collaborative learning. By exploiting the natural signal superposition property of wireless channels, OTA computation enables multiple devices to simultaneously transmit their analog model updates, which are aggregated directly in the air at the AP. Compared to conventional orthogonal transmission schemes, OTA-based FL (OTA-FL) can dramatically reduce communication latency and bandwidth consumption, making it a promising enabler of scalable edge intelligence \cite{10388035, 18952884, 10988827}. However, in these practical wireless environments, OTA-FL remains highly sensitive to channel impairments such as fading, noise, and multiuser interference, which can distort the aggregated signal and degrade learning performance. Additionally, unfavorable channel conditions may prevent some devices from reliably participating in aggregation, further limiting system efficiency \cite{10932699, 9738803}.

Recognizing the pivotal role of the wireless propagation environment, recent research efforts have focused on enhancing OTA-FL performance through advanced antenna and propagation technologies. Multiple-input multiple-output (MIMO) techniques exploit spatial degrees of freedom to enable beamforming and interference management, thereby improving aggregation reliability~\cite{10261444, 10444710, 10381609}. Directional antenna systems similarly enhance link quality by concentrating transmit or receive power along preferred spatial directions~\cite{10258328, 9849110,10906511}. However, the performance of such approaches is inherently constrained by fixed radiation patterns and limited adaptability to rapid small-scale fading variations, particularly in dynamic wireless environments. In parallel, reconfigurable intelligent surfaces (RIS) have been proposed to actively manipulate the wireless propagation environment via programmable reflection coefficients~\cite{9829190, 10363658, 10912454, 11342412, 10529197, 10881858}. By reshaping the channel conditions, RIS can significantly improve link reliability and mitigate channel impairments. Nevertheless, these gains typically come at the cost of additional infrastructure, centralized control, and increased channel estimation and signaling overhead. These limitations motivate the exploration of more flexible and infrastructure-light solutions that can adapt to instantaneous channel conditions without relying on external environmental control.

In this context, fluid antennas (FA) have recently emerged as a novel and transformative antenna technology capable of overcoming these limitations~\cite{9715064, 10146274, 11175437, 11106811, pakravan2026fluid1, 10753482}. Unlike conventional fixed-position antennas (FPA), FA allows the radio-frequency (RF) front-end to dynamically relocate within a compact spatial region, thereby exploiting position-domain reconfigurability to enhance spatial diversity \cite{skouroumounis2024swipt}. By selecting the antenna position experiencing the most favorable channel realization, FA can enhance link reliability, user participation, and overall OTA-FL performance without additional antennas or infrastructure, which makes it particularly appealing for dense uplink edge-intelligence scenarios.

Motivated by these advantages, several recent studies have explored the integration of FA technology into OTA-FL systems \cite{saeidpwcnc,ahmadzadeh2025ai,ahmadzadeh2025advancing, saeidp, zhao2025fluid}. Existing works primarily focus on deploying FA at the AP and optimizing antenna positions to maximize the number of participating devices under a prescribed mean squared error (MSE) constraint or to minimize the learning optimality gap. While these contributions demonstrate FA’s potential, they rely heavily on optimization-based frameworks and numerical evaluations. Notably, a rigorous analytical characterization of the fundamental performance limits of FA-aided OTA-FL systems remains largely unexplored.

One of the key challenges in analyzing FA-enabled systems stems from the strong spatial correlation among antenna ports caused by their close physical proximity \cite{ghadi2024physical}. Such correlation significantly alters the statistical behavior of the wireless channel and, consequently, the system performance. Classical spatial correlation models, including those derived from Jakes’ framework \cite{9998496, 10623405, 10716282}, often lead to analytically intractable expressions, limiting their usefulness for theoretical performance analysis. As a result, the impact of spatially correlated fading on OTA-FL aggregation accuracy and device participation has not yet been explicitly quantified.

In this paper, we aim to bridge these gaps by presenting a comprehensive analytical study of an FA-assisted OTA-FL system, where each user is equipped with a single FA while the AP employs a conventional FPA. In contrast to existing works that primarily rely on optimization-based or learning-based designs, our focus is on understanding how FA mobility enhances OTA-FL performance and how spatial correlation among FA ports influences aggregation reliability and user participation. Unlike RIS-assisted FL approaches~\cite{9625822, 9626135, 10649032}, which rely on environment reconfiguration and typically involve joint optimization of reflecting elements and transmission strategies, the proposed framework leverages position-domain reconfigurability at the user side and enables a tractable analytical characterization of aggregation performance. To achieve analytical tractability while capturing realistic dependence structures, we adopt a copula-based statistical framework, specifically, the Clayton copula, to model spatial correlation across FA ports in a tractable manner. This approach enables closed-form performance analysis while capturing realistic dependence structures among fading coefficients and provides new insights into the fundamental benefits and limitations of FA-assisted OTA-FL.
To further highlight these differences, Table~\ref{tab:lit_comparison} summarizes the key distinctions between representative prior works and the proposed framework.

\begin{table*}[t]
\centering
\caption{{Comparison of the proposed work with representative prior studies}}
\label{tab:lit_comparison}
\renewcommand{\arraystretch}{1.2}
\setlength{\tabcolsep}{3.5pt}

{
\begin{tabular}{|p{0.55cm}|p{2.56cm}|c|c|c|c|p{5.2cm}|p{2.55cm}|c|}
\hline
\multirow{2}{*}{\textbf{Ref.}} & \multirow{2}{*}{\textbf{Performance metrics}} & \multicolumn{4}{c|}{\textbf{System model}} & \multirow{2}{*}{\textbf{Main contribution}} & \multirow{2}{*}{\textbf{Proposed method}} \\
\cline{3-6}
 &  & \textbf{OTA} & \textbf{FL} & \textbf{FA deployment} & \textbf{Spatial correlation} &  &  \\
\hline

\cite{saeidpwcnc} & Optimality gap, Training loss / accuracy & \checkmark & \checkmark & AP-side FA & $\times$ & Develops an FA-assisted OTA-FL framework with rigorous convergence analysis & Learning/optimization-based framework \\
\hline

\cite{ahmadzadeh2025ai} & User selection rate, learning accuracy  & \checkmark & \checkmark  & AP-side FA & $\times$  & Proposes an FA-assisted OTA-FL framework with joint optimization of antenna positions, beamforming, and user selection under MSE constraints & Learning/optimization-based framework \\
\hline

\cite{ahmadzadeh2025advancing} & Optimality gap, MSE & \checkmark & \checkmark & AP-side FA& $\times$ & Proposes a unmanned aerial vehicle assisted OTA-FL framework with movable antennas
& Learning/optimization-based framework \\
\hline

\cite{saeidp} & MSE & \checkmark  & $\times$ & AP-side FA & $\times$ & Proposes a FA array-enabled OTA computation framework with robust resource allocation under channel uncertainty& Block coordinate descent method\\
\hline

\cite{zhao2025fluid} & Convergence rate, training loss & \checkmark & \checkmark & AP-side FA & $\times$ & Proposes a FA-enabled OTA-FL framework with joint optimization of device selection, receiver beamforming, and antenna positioning under time-varying channels & Penalty dual decomposition with successive convex approximation \\
\hline

\textbf{This work} & \textbf{MSE + participation statistics (CDF/PMF)} & \checkmark & \checkmark & \textbf{User-side single FA} & \checkmark & \textbf{Proposes an OTA-FL framework that enhances aggregation reliability and user participation, supported by analytical performance characterization under correlated fading} & \textbf{Closed-form analysis + simulation + learning experiments} \\
\hline
\end{tabular}
}

\end{table*}

The main contributions of this paper are summarized as follows:
\begin{itemize}
    \item We develop a comprehensive analytical framework for FA-assisted OTA-FL and derive tractable expressions for key learning performance metrics, including an upper bound on the aggregation MSE and the expected number of participating FL devices.

    \item We model spatially correlated fading across FA ports using a copula-based approach and characterize the effective FA channel gain under correlated Rayleigh fading, providing new insights into the role of spatial correlation.

      \item Through numerical simulations and learning experiments on a real-world dataset, we validate the analytical results and demonstrate the significant performance gains achievable by integrating FAs into OTA-FL systems, highlighting the effects of FA mobility and port correlation on global model accuracy and loss.
\end{itemize}

The remainder of this paper is organized as follows. Section~II describes the system model, including the FL framework, FA communication model, and aggregation error formulation. Section~III presents the analytical performance evaluation of FA-assisted OTA-FL in terms of MSE minimization and device selection maximization. Section~IV evaluates the impact of FA spatial correlation on MSE and the number of participating users. Numerical and learning performance results are provided in Section~V, followed by concluding remarks in Section~VI.

{For clarity, the main symbols and notations used throughout the paper are summarized in Table~\ref{tab:notation}.}

\begin{table}[t]
\centering
\caption{ Main symbols and notations}
\label{tab:notation}
\renewcommand{\arraystretch}{1.2}
\setlength{\tabcolsep}{4pt}

{
\begin{tabular}{|c|p{6.5cm}|}
\hline
\textbf{Symbol} & \textbf{Description} \\
\hline

$K$ & Number of user devices (UDs) \\
\hline

$\mathcal{K}$ & Set of all user devices \\
\hline

$\mathcal{S}$ & Set of participating users  \\
\hline

$|\mathcal{S}|$ & Number of participating users \\
\hline

$N$ & Number of FA ports \\
\hline

$W$ & Normalized antenna aperture size \\
\hline

$\lambda$ & Carrier wavelength \\
\hline

$h_{k,n}$ & Channel coefficient between UD $k$ and AP at FA port $n$ \\
\hline

$h_k$ & Effective channel after FA port selection \\
\hline

$p_k$ & Transmit scaling factor of UD $k$ \\
\hline

$p_{\max}$ & Maximum transmit power constraint \\
\hline

$\boldsymbol{u}_k$ & Local model update vector of UD $k$ \\
\hline

$\boldsymbol{y}$ & Received signal at the AP \\
\hline

$\boldsymbol{z}$ & Additive white Gaussian noise (AWGN) \\
\hline

$\sigma^2$ & Noise variance \\
\hline

$\eta$ & Denoising factor \\
\hline

$\boldsymbol{e}$ & Aggregation error vector \\
\hline

$\tau$ & MSE threshold for user participation \\
\hline

$\delta$ & Channel gain threshold for user selection \\
\hline

$\Gamma$ & Number of participating users \\
\hline

$\beta$ & Copula dependence parameter \\
\hline

$\Theta_k$ & Transformed variable $(p_{\max}|h_k|^2)^{-1}$ \\
\hline

$F_{|h_k|^2}(x)$ & CDF of effective channel gain \\
\hline

$\boldsymbol{w}$ & Global model parameter \\
\hline

$F(\boldsymbol{w})$ & Global loss function \\
\hline

$\gamma$ & Learning rate \\
\hline

$\mu$ & Polyak--Łojasiewicz (PL) constant \\
\hline

$L$ & Smoothness constant of the loss function \\
\hline

$\sigma_g^2$ & Variance of stochastic gradients \\
\hline

$\kappa$ & Upper bound on gradient norm \\
\hline

\end{tabular}
}

\end{table}

\section{System Model}

We consider an uplink OTA-FL system consisting of $K$ user devices (UDs), indexed by $\mathrm{UD}_k$, $k \in \mathcal{K} \triangleq \{1,2,\ldots,K\}$, and a single AP. Each UD is equipped with a single FA, while the AP employs a conventional FPA. The UDs collaboratively train a global learning model hosted at the AP by periodically exchanging model updates through OTA aggregation. The overall architecture of the proposed FA-assisted OTA-FL system is illustrated in Fig.~\ref{fig:system_model}.

\begin{figure}[!t]
	\centering
	\includegraphics[width=8cm, height=4.2cm]{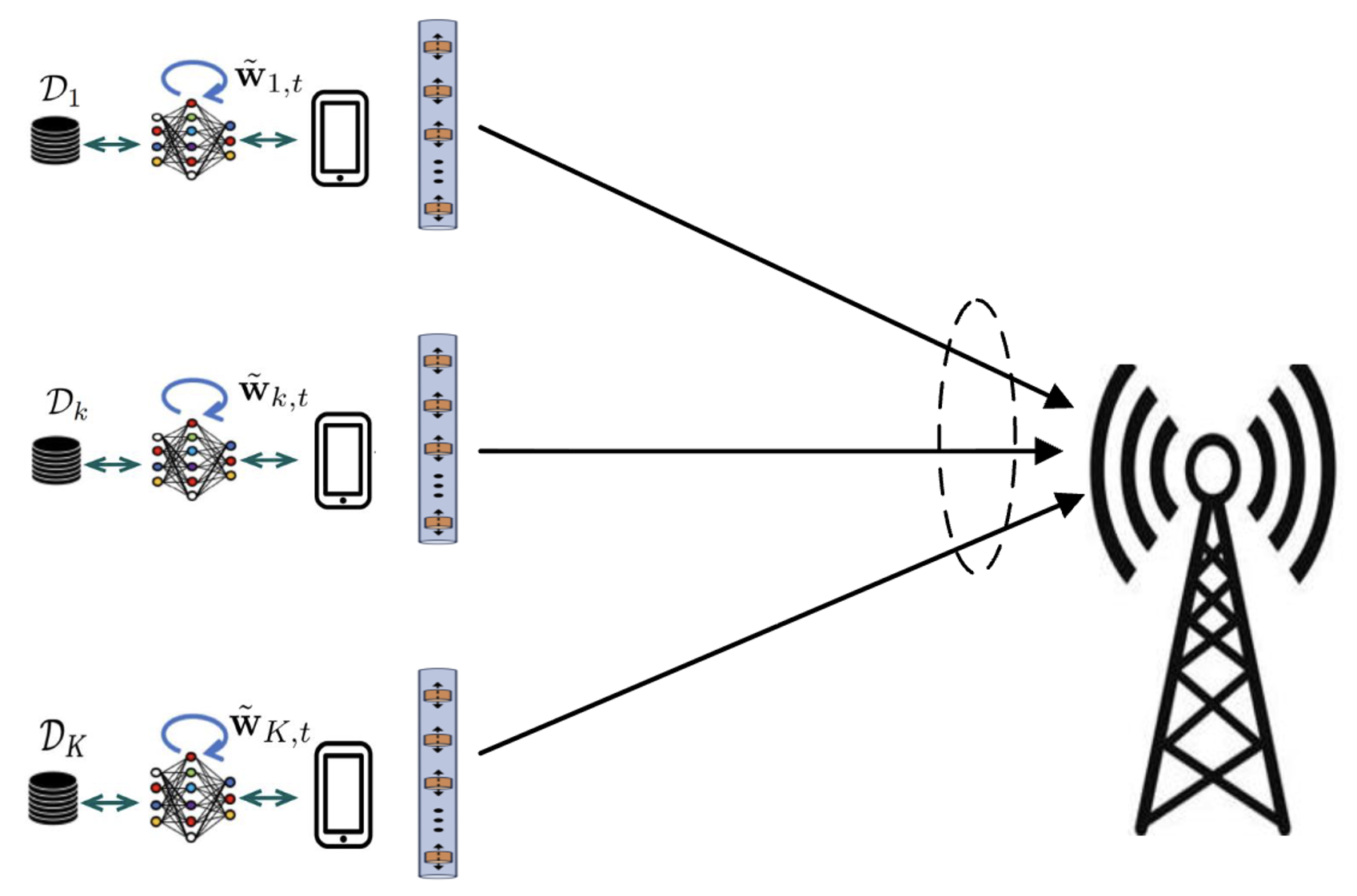}
	\caption{An illustration of the proposed FA-assisted OTA-FL system.}
	\label{fig:system_model}
\end{figure}

\subsection{FL Model}

The objective of FL is to minimize a global empirical loss function defined as the average of the local loss functions across all participating UDs. Specifically, at the \(t\)-th communication round, the global optimization problem is given by
\begin{equation}
\label{goal}
\min_{\boldsymbol{w}_t} F(\boldsymbol{w}_t)
=
\min_{\boldsymbol{w}_t}
\frac{1}{K}
\sum_{k=1}^{K}
F_k(\boldsymbol{w}_t),
\end{equation}
where \(\boldsymbol{w}_t \in \mathbb{R}^{d}\) denotes the global model parameter vector of dimension \(d\),and $F_k(\boldsymbol{w}_t)$ represents the local empirical loss function at $\mathrm{UD}_k$
, evaluated using its locally available data samples.

To solve \eqref{goal}, the standard federated averaging (FedAvg) algorithm~\cite{18952884} is adopted over \(T\) communication rounds. Each training round consists of the following three stages.

\textbf{1) User selection:}
At the beginning of the \(t\)-th round, the AP selects a subset of UDs \(\mathcal{S}_t \subseteq \mathcal{K}\) to participate in the training process and broadcasts the current global model parameters \(\boldsymbol{w}_t\) to the selected devices.

\textbf{2) Local training:}
Each selected UD performs local model optimization using stochastic gradient descent (SGD) based on a randomly sampled minibatch of its locally stored samples. Specifically, the local model update at \(\mathrm{UD}_k\) is given by
\begin{equation}
\boldsymbol{w}_{k,t}
= \boldsymbol{w}_t
- \gamma \nabla \hat{F}_k(\boldsymbol{w}_t;\zeta_k),
\end{equation}
where \(\gamma \in (0,1)\) denotes the learning rate and \(\hat{F}_k(\boldsymbol{w}_t;\zeta_k)\) is the empirical loss evaluated over the minibatch \(\zeta_k\). The extension to multiple local update steps is straightforward.

\textbf{3) Model aggregation:}
After completing local training, the selected UDs transmit their updated model parameters to the AP. The AP then aggregates the received updates to form the global model for the next round as
\begin{equation}
\label{modelaggregation}
\boldsymbol{w}_{t+1}
= \frac{1}{|\mathcal{S}_t|}
\sum_{k \in \mathcal{S}_t}
\boldsymbol{w}_{k,t}.
\end{equation}

The above procedure is iteratively executed until convergence is achieved or a predefined maximum number of training rounds \(T\) is reached.

\subsection{Communication Model}

We focus on the uplink OTA aggregation phase, where all selected UDs transmit simultaneously over a shared wireless channel.

Each UD employs a single-RF-chain FA with $N$ predefined antenna ports distributed over a linear aperture of length $W\lambda$, where $\lambda$ is the carrier wavelength. The relative displacement of the $n$-th FA port is given by
\begin{equation}
d_n
= \frac{n-1}{N-1} W\lambda,
\qquad n = 1,2,\ldots,N.
\end{equation}

To exploit position-domain diversity, each UD dynamically selects the FA port that maximizes the instantaneous channel power gain. We assume that the AP has perfect knowledge of the instantaneous fading coefficients for all FA ports. The resulting effective uplink channel gain for $\mathrm{UD}_k$ is therefore defined as
\begin{equation}
\label{eq:effective_channel}
|h_k|^{2}
\triangleq
\max
\big\{
|h_{k,1}|^{2}, |h_{k,2}|^{2}, \ldots, |h_{k,N}|^{2}
\big\},
\end{equation}
where $h_{k,n}$ denote the complex small-scale fading coefficient between $\mathrm{UD}_k$ and the AP when the $n$-th FA port is activated. 

Due to the compact spatial arrangement of FA ports, the channel coefficients are generally spatially correlated. Under a two-dimensional isotropic scattering environment, the spatial correlation between FA ports follows Jakes' model~\cite{ghadi2023fluid} and is expressed as
\begin{equation}
\mathbb{E}\!\left[h_{k,i} h_{k,j}^{\ast}\right]
=
\delta^{2}
J_{0}
\!\left(
2\pi
\frac{|i-j|}{N-1}
W
\right),
\end{equation}
where $\delta^{2}$ denotes the large-scale fading power and $J_{0}(\cdot)$ is the zeroth-order Bessel function of the first kind. \footnote{{The adopted model captures spatial correlation among FA ports but does not explicitly account for electromagnetic mutual coupling. This abstraction is used to preserve analytical tractability. Incorporating mutual coupling into the FA channel model, particularly for dense port configurations, is an important direction for future work.}}

Each selected UD encodes its local model update into a transmit symbol vector \(\boldsymbol{u}_k
\triangleq
\boldsymbol{w}_{k,t}\), which is normalized such that
\(\mathbb{E}\!\left[\boldsymbol{u}_k \boldsymbol{u}_k^{H}\right] = \mathbf{I}_d\).
The UD then activates the FA port corresponding to the strongest instantaneous channel realization and applies a power-scaling factor to facilitate OTA aggregation. For notational simplicity, the training round index \(t\) is omitted in the following analysis.

Accordingly, the received signal at the AP is expressed as
\begin{equation}
\boldsymbol{y}
=
\sum_{k \in \mathcal{S}}
p_k \, h_{k} \, \boldsymbol{u}_k
+
\boldsymbol{z},
\end{equation}
where \(p_k\) denotes the transmit power-scaling factor of \(\mathrm{UD}_k\), and $\boldsymbol{z} \in \mathbb{C}^{1\times d}$ is an additive white Gaussian noise (AWGN) matrix whose entries are independently and identically distributed as $\mathcal{CN}(0,\sigma^2)$. Additionally, we assume that the additive noise $\boldsymbol{z}$ is independent of the fading coefficients, which ensures that the aggregation error arises solely from channel distortions and power scaling mismatches.

Each UD is subject to a maximum transmit power constraint, given by
\begin{equation}
\label{power}
\frac{1}{d}\, p_k^{2}
\le p_{\max},
\qquad \forall k \in \mathcal{K}.
\end{equation}

Consistent with the ideal model aggregation, the desired aggregated signal at the AP is defined as
\(\boldsymbol{y}_{\mathrm{tar}} = \sum_{k \in \mathcal{S}} \boldsymbol{u}_k\), where the final normalization by $|\mathcal{S}|$ is performed at the AP to recover the arithmetic mean in \eqref{modelaggregation}.
Based on the received signal, the AP constructs an estimate of the aggregated model as
\begin{equation}
\label{rt}
\boldsymbol{y}_{\mathrm{est}}
=
\frac{1}{\sqrt{\eta}}
\sum_{k \in \mathcal{S}}
p_k
h_{k}
\boldsymbol{u}_k
+
\frac{1}{\sqrt{\eta}} \boldsymbol{z}.
\end{equation}
As a result, the OTA aggregation error is defined as
\begin{equation}
\label{eq:aggregation_error}
\boldsymbol{e}
\triangleq
\boldsymbol{y}_{\mathrm{tar}}
-
\boldsymbol{y}_{\mathrm{est}}
=
\sum_{k \in \mathcal{S}}
\left(
1
-
\frac{p_k h_k}{\sqrt{\eta}}
\right)
\boldsymbol{u}_k
-
\frac{1}{\sqrt{\eta}}
\boldsymbol{z}.
\end{equation}

The design of the transmit scaling factors $\{p_k\}$ and the denoising parameter $\eta$, which jointly govern the aggregation MSE and the number of participating UDs, will be addressed in the subsequent section.

\section{Performance Analysis}

In this section, we analytically investigate the impact of FAs on the learning performance of OTA-FL systems. Following \cite{ahmadzadeh2025ai}, the learning performance is characterized in terms of the optimality gap, defined as the difference between the achieved global loss and the minimum attainable loss of the underlying learning problem.

In OTA-FL, the convergence behavior is fundamentally governed by two key factors: i) the aggregation accuracy at the AP, typically quantified by the aggregation MSE, and ii) the number of UDs participating in each communication round. {These factors capture the effects of communication distortion and statistical efficiency, respectively.}

{To explicitly establish this connection, we adopt the optimality-gap framework in \cite{ahmadzadeh2025ai} and specialize it to the considered FA-assisted system. Let $\mathrm{MSE}_t$ denote the aggregation MSE and $\mathcal{S}_t$ the set of participating users at communication round $t$.}\footnote{{The result follows standard convergence analysis for federated optimization under smoothness and the Polyak--Łojasiewicz condition; see \cite{ahmadzadeh2025ai} for details.}}

{\begin{theorem}
\label{thm:convergence}
Under standard assumptions for federated optimization, the optimality gap after $T$ communication rounds satisfies
\begin{equation}
\small
\mathbb{E}[F(\boldsymbol{w}_{T+1})] - F(\boldsymbol{w}^{\ast})
\leq
\psi^{T} \!\left(\mathbb{E}[F(\boldsymbol{w}_1)] - F(\boldsymbol{w}^{\ast})\right)
+
\sum_{t=1}^{T} \psi^{T-t} \Delta_t,
\end{equation}
where $\psi = 1 - \gamma \mu$, with $\gamma$ denoting the learning rate and $\mu$ the Polyak--Łojasiewicz constant. The residual term is given by
\begin{equation}
\Delta_t =
2 \gamma \kappa \left(1 - \frac{|\mathcal{S}_t|}{K} \right)^2
+
\frac{\gamma^2 L}{|\mathcal{S}_t|^2}
\sum_{k \in \mathcal{S}_t} \frac{\sigma_g^2}{|\zeta_k|}
+
\frac{L}{2} \, \mathrm{MSE}_t,
\end{equation}
where $L$ is the smoothness constant of the loss function, $\sigma_g^2$ is the gradient variance, and $\kappa$ bounds the gradient norm.\footnote{{For notational simplicity, the time index $t$ is omitted in the following.}}
\end{theorem}}

{\begin{proof}
The result follows by adapting the convergence analysis in \cite{ahmadzadeh2025ai}. The detailed derivation is omitted for brevity.
\end{proof}}

{Theorem~\ref{thm:convergence} shows that the convergence behavior is jointly determined by the aggregation distortion and the number of participating users. In particular, selecting the learning rate such that $|\psi|<1$ ensures that the first term decays with $T$, and the optimality gap is primarily governed by the residual term $\Delta_t$.}

{Since $\Delta_t$ increases with the aggregation MSE and decreases with the number of participating users, improving convergence requires reducing the aggregation error while maintaining a sufficiently large participation set. This observation provides the theoretical foundation for the analysis developed in the following subsections.}

{Motivated by this observation, the considered FA-assisted OTA-FL system can be interpreted through two coupled design objectives. For a given user set $\mathcal{S}$, the first objective is to minimize the aggregation MSE by properly designing the transmit scaling factors and the denoising factor under individual power constraints:
\begin{equation}
\label{11}
\begin{aligned}
\min_{\{p_k\},\,\eta} \quad & \mathrm{MSE} \\
\text{s.t.} \quad & \frac{1}{d}p_k^2 \leq p_{\max}, \;\forall k \in \mathcal{S}.
\end{aligned}
\end{equation}}

{Based on the resulting MSE expression, the second objective is to maximize the number of participating users subject to a target MSE constraint:
\begin{equation}
\label{12}
\begin{aligned}
\max_{\mathcal{S}\subseteq\mathcal{K}} \quad & |\mathcal{S}| \\
\text{s.t.} \quad & \mathrm{MSE} \leq \tau.
\end{aligned}
\end{equation}}

These formulations show explicitly the intrinsic trade-off between aggregation accuracy and user participation, which directly governs the convergence performance of OTA-FL. In the following, we first solve \eqref{11} and then characterize the feasible user set induced by \eqref{12}.

\subsection{Aggregation MSE Minimization}

In OTA-FL, global model aggregation is performed via analog OTA computation, where the objective is to minimize the discrepancy between the desired aggregated model update and its estimate at the AP. This discrepancy is quantified by the MSE of the aggregation error vector \(\boldsymbol{e}\) at a given training round, as defined in~\eqref{eq:aggregation_error}. Specifically, the aggregation MSE is given by
\begin{equation}
\label{mse}
\mathrm{MSE}
= \mathbb{E}\!\left[ \left\| \boldsymbol{e} \right\|^2 \right] \\
=
\sum_{k \in \mathcal{S}}
\left|
1
-
\frac{p_k h_{k}}{\sqrt{\eta}}
\right|^2
+
\frac{d}{\eta}\sigma^2.
\end{equation}

To minimize the MSE, the transmit power-scaling coefficients \(\{p_k\}\) are designed according to a zero-forcing (ZF) amplitude inversion strategy~\cite{chen2023joint}, which compensates for channel distortions by aligning the effective channel gains of all participating UDs. Specifically, the optimal transmit scaling factor for \(\mathrm{UD}_k\) is given by
\begin{equation}
\label{pk}
p_k
=
\sqrt{\eta}\,
\frac{h_{k}^{H}}{|h_{k}|^{2}},
\qquad \forall k \in \mathcal{S},
\end{equation}
which ensures coherent aggregation of the transmitted model updates at the AP.

To satisfy the individual transmit power constraint in~\eqref{power}, the denoising factor \(\eta\) must satisfy
\begin{equation}
\label{eta}
\eta
\le
d\, p_{\max} |h_{k}|^{2},
\qquad \forall k \in \mathcal{S}.
\end{equation}
Since \(\eta\) must be chosen uniformly for all participating UDs, it is determined by the UD experiencing the weakest effective channel gain.

Substituting~\eqref{pk} and~\eqref{eta} into~\eqref{mse}, the resulting aggregation MSE simplifies to
\begin{equation}
\label{lll}
\mathrm{MSE}
=
\frac{\sigma^2}{p_{\max}}
\max_{k \in \mathcal{S}}
\frac{1}{|h_{k}|^{2}}.
\end{equation}

This equality provides an important insight into the behavior of OTA-FL systems. Specifically, the aggregation MSE is determined by the user with the weakest effective channel gain, i.e., the minimum $|h_k|^2$ among all participating users. This reflects a bottleneck effect inherent to channel inversion–based OTA aggregation, where all users must align their signals, and the transmit power is limited by the worst channel condition.

As a result, even if most users experience strong channel gains, the overall aggregation accuracy is constrained by the weakest user. This highlights the critical importance of improving the minimum effective channel gain across users. In this context, FA-enabled users can exploit position-domain diversity to select the best antenna port, thereby increasing their effective channel gain. Consequently, FA selection reduces the likelihood of extremely weak channels, leading to lower aggregation MSE and enabling more users to satisfy the participation condition.

\subsection{Participating User Maximization}

Beyond aggregation accuracy, the convergence behavior of OTA-FL is strongly influenced by the number of participating UDs in each training round. As shown in~\cite{ahmadzadeh2025ai}, the optimality gap decreases with the number of participating users and increases with the aggregation MSE. Consequently, improving convergence requires maximizing the number of participating UDs while ensuring that the aggregation MSE does not exceed a prescribed threshold \(\tau\).

Based on the MSE expression in~\eqref{lll}, the set of UDs eligible to participate in a given training round is determined as
\begin{equation}
\label{userpolicy1}
\begin{aligned}
\mathcal{S}
&=
\left\{
k \in \mathcal{K}
\;\middle|\;
\mathrm{MSE} \le \tau
\right\}
\\
&\overset{(a)}{=}
\left\{
k \in \mathcal{K}
\;\middle|\;
\frac{\sigma^2}{p_{\max}}
\max_{k \in \mathcal{S}}
\frac{1}{|h_{k}|^{2}}
\le \tau
\right\}
\\
&\overset{(b)}{=}
\left\{
k \in \mathcal{K}
\;\middle|\;
|h_{k}|^{2} \ge \delta
\right\}
\\
&\overset{(c)}{=}
\left\{
k \in \mathcal{K}
\;\middle|\;
\max \left\{
|h_{k,1}|^{2}, |h_{k,2}|^{2}, \ldots, |h_{k,N}|^{2}
\right\}
\ge \delta
\right\},
\end{aligned}
\end{equation}
where \((a)\) follows from \eqref{lll}, \((b)\) is obtained by defining the threshold
\(\delta \triangleq \frac{\sigma^2}{\tau p_{\max}}\), and \((c)\) follows from the effective FA channel gain definition in~\eqref{eq:effective_channel}.
The proposed FA port selection assumes the availability of instantaneous channel state information (CSI) across all candidate ports. In practice, this requires pilot-based estimation, where the overhead scales linearly with the number of users and FA ports, i.e., $KN\tau_p$ per communication round, with $\tau_p$ denoting the pilot duration per port. 
This additional overhead increases the round duration and may become non-negligible in fast-fading environments. Under a fixed training time budget, it reduces the number of effective communication rounds, thereby reducing the convergence speed in real time. Therefore, the proposed framework is most suitable for quasi-static or moderately time-varying channels, while low-overhead CSI acquisition for fast-fading scenarios remains an important direction for future work.

\textbf{Remark 1:}
Equations~\eqref{lll} and~\eqref{userpolicy1} highlight the pivotal role of the effective channel gain $|h_{k}|^{2}$ in determining both aggregation accuracy and user participation. FA-enabled UDs exploit position-domain diversity by activating the port with the strongest instantaneous channel realization, thereby reducing aggregation MSE and increasing the number of eligible users per round. However, strong spatial correlation among FA ports reduces channel gain diversity and limits these benefits, motivating the statistical characterization of FA spatial correlation presented in the next section.

\textbf{Remark 2:}
The user selection policy in \eqref{userpolicy1} is based on instantaneous channel conditions, which raises potential fairness concerns across multiple communication rounds. In particular, users experiencing persistently weak channel realizations may be selected less frequently, leading to an imbalance in long-term participation. Nevertheless, under the adopted block-fading model, channel realizations vary randomly across rounds, which provides an inherent degree of temporal fairness, as users may experience favorable channel conditions at different times. In addition, the use of FAs enhances fairness by allowing each user to select the antenna position with the strongest instantaneous channel gain, thereby increasing its likelihood of meeting the participation threshold.\footnote{Despite these advantages, ensuring strict fairness or guaranteeing uniform participation across users requires additional scheduling or resource allocation mechanisms, which are beyond the scope of this work and constitute an interesting direction for future research.}

For clarity, the main steps of the proposed FA-assisted OTA-FL procedure are summarized in Algorithm~\ref{alg:fa_otafl}.\footnote{For each communication round, the proposed scheme estimates the channel gains over all $N$ FA ports for each of the $K$ users and selects the strongest port, resulting in a complexity of $\mathcal{O}(KN)$. The complexity of local model training is model-dependent and is not included in this analysis.}

\begin{algorithm}[t]
\caption{Proposed FA-Assisted OTA-FL Procedure}
\label{alg:fa_otafl}
\begin{algorithmic}[1]
\Require User set $\mathcal{K}$, FA ports $N$, power constraint $p_{\max}$, MSE threshold $\tau$
\State Initialize global model $\boldsymbol{w}_1$
\For{each communication round $t=1,\ldots,T$}
    \State Broadcast $\boldsymbol{w}_t$ to all users
    \For{each user $k \in \mathcal{K}$}
        \State Estimate $\{|h_{k,n}|^2\}_{n=1}^N$
        \State Select FA port and obtain effective gain using (5)
    \EndFor
    \State Determine participating set $\mathcal{S}$ using (19)
    \State Set denoising factor $\eta$ using (17)
    \For{each $k \in \mathcal{S}$}
        \State Perform local update
        \State Compute $p_k$ using (16)
    \EndFor
    \State OTA aggregation and update $\boldsymbol{w}_{t+1}$
\EndFor
\end{algorithmic}
\end{algorithm}

\section{Theoretical Analysis}

In this section, we present a rigorous theoretical characterization of the FA-assisted OTA-FL system introduced in the previous sections. Our objective is to analytically quantify the impact of FA port selection and spatial correlation on both aggregation accuracy and user participation. Specifically, we derive closed-form expressions for the distribution of the aggregation MSE and the number of selected UDs, while explicitly accounting for correlated fading across FA ports. 

Furthermore, we characterize the distribution of the effective FA channel gain resulting from port selection. This analysis provides insights into how spatial correlation fundamentally constrains the diversity gains achievable through FA-based port selection, and consequently, how it impacts the overall learning performance of OTA-FL systems.

\subsection{Statistical Characterization of OTA-FL}

To evaluate the impact of FA-assisted transmission on OTA-FL performance, we focus on minimizing the aggregation error for a given number of selected devices, denoted by \(S = |\mathcal{S}|\). From \eqref{lll}, the normalized aggregation MSE can be expressed as
\begin{equation}
\label{newMSE}
\frac{\mathrm{MSE}}{\sigma^{2}}
= \min_{\mathcal{S} \subseteq \mathcal{K}} \; \max_{k \in \mathcal{S}}
\frac{1}{p_{\max} |h_{k}|^{2}}.
\end{equation}

The fading coefficients are modeled as block-fading variables that are independently drawn across communication rounds. Consequently, any function of these coefficients such as the aggregation MSE is a random variable (RV). 

To facilitate distributional analysis, we define the RV
\(\Theta_k \triangleq (p_{\max} |h_{k}|^{2})^{-1}\). Then, its cumulative distribution function (CDF) can be calculated as follows:
\begin{equation}
\label{FKXK}
\begin{aligned}
F_{\Theta_k}(\theta_k)
&= \Pr(\Theta_k < \theta_k) = \Pr(|h_{k}|^{2} > (p_{\max} \theta_k)^{-1}) \\
&= 1 - \Pr(|h_{k}|^{2} < (p_{\max} \theta_k)^{-1}).
\end{aligned}
\end{equation}

Based on \eqref{newMSE}, the normalized MSE can be written as
\begin{equation}
\gamma \triangleq \frac{\mathrm{MSE}}{\sigma_n^2} 
= \min_{\mathcal{S} \subseteq \mathcal{K}} \; \max_{k \in \mathcal{S}} \Theta_k.
\end{equation}

To ensure that the normalized aggregation MSE remains below a prescribed threshold \(\tau\), we evaluate the probability \(\Pr(\gamma < \tau)\). Let the RVs \(\{\Theta_k\}_{k=1}^{K}\) be ordered as
\(\Theta_{1:K} \le \Theta_{2:K} \le \cdots \le \Theta_{K:K}\).
Under this ordering, \(\gamma\) corresponds to the \(|\mathcal{S}|\)-th order statistic, i.e., \(\gamma = \Theta_{|\mathcal{S}|:K}\). Accordingly, the CDF of \(\gamma\) is expressed as
\begin{equation}
\label{disMSE}
\begin{aligned}
&F_{\gamma}(\tau)
= \Pr(\gamma < \tau) = 1 - \Pr\!\left(\Theta_{|\mathcal{S}|:K} \geq \tau\right) \\
&\overset{(a)}{=} 1 - \sum_{i=0}^{|\mathcal{S}|-1} \binom{K}{i}
\big[\Pr\!\left(\Theta_i < \tau\right)\big]^{i}
\big[\Pr\!\left(\Theta_i > \tau\right)\big]^{\,K-i}\\
&\overset{(b)}{=} 1 - \sum_{i=0}^{|\mathcal{S}|-1} \binom{K}{i}
\big[1 - \Pr\!\left(|h_{k}|^{2} < (p_{max} \tau)^{-1}\right)\big]^{i} \\&\hspace{2cm}
\big[\Pr\!\left(|h_{k}|^{2} < (p_{max} \tau)^{-1}\right)\big]^{\,K-i},
\end{aligned}
\end{equation}
where \((a)\) follows from order-statistics theory and a binomial counting argument, and \((b)\) follows from \eqref{FKXK}.

For the user selection policy in \eqref{userpolicy1}, the participating set can be equivalently written as
\begin{equation}
	\label{userpolicy}
	\begin{aligned}
		\mathcal{S} 
		&\overset{}{=} \left\{ k \in \mathcal{K} \;\middle|\; (p_{\max}|h_{k}|^2)^{-1} \le {\tau}/{\sigma^2} \right\}.\\
		&\overset{}{=} \left\{ k \in \mathcal{K} \;\middle|\; \Theta_k \le {\tau}/{\sigma^2} \right\}.
	\end{aligned}
\end{equation}
Let \(\Gamma = |\mathcal{S}|\) denote the number of selected UDs. Modeling \(\Gamma\) as a binomial RV, its probability mass function (PMF) is expressed as:
\begin{equation}
\label{disuser}
\begin{aligned}
&\Pr\!\left(\Gamma = |\mathcal{S}|\right)\\&
=\binom{K}{|\mathcal{S}|}
   \big[\Pr(\Theta_k \le {\tau}/{\sigma^2} )\big]^{|\mathcal{S}|}
   \big[\Pr(\Theta_k > {\tau}/{\sigma^2} )\big]^{K-|\mathcal{S}|}
   \\&
=\binom{K}{|\mathcal{S}|}
   \big[1 - \Pr\!\left(|h_{k}|^{2} < \sigma^2(p_{max} \tau)^{-1}\right)\big]^{|\mathcal{S}|} \\& \hspace{2cm}
   \times \big[\Pr\!\left(|h_{k}|^{2} < \sigma^2(p_{max} \tau)^{-1}\right) \big]^{K-|\mathcal{S}|}. 
   \end{aligned}
\end{equation}

\textbf{Remark 3:} 
Equations~\eqref{disMSE} and~\eqref{disuser} demonstrate that both the aggregation MSE outage probability and the user participation statistics are fundamentally dictated by the distribution of the effective FA channel gain \(|h_{k}|^{2}\). Consequently, an accurate characterization of this distribution, particularly under spatial correlation, is essential for precise evaluation of learning performance in FA-assisted OTA-FL systems.

\subsection{Statistical Characterization of the Effective Channel Gain}

In this subsection, we derive the CDF of the effective FA channel gain \(|h_{k}|^{2}\), defined in \eqref{eq:effective_channel}, while explicitly accounting for spatial correlation among FA ports. 

To model statistical dependence among FA ports, we adopt a copula-based framework grounded in Sklar’s theorem.

\begin{theorem}[Sklar’s Theorem~\cite{nelsen2006introduction}]
Let \(\mathbf{X} = (X_1,\ldots,X_J)\) be a random vector with joint CDF \(F_{\mathbf{X}}\) and marginal CDFs \(F_{X_j}\). Then, there exists a copula \(C:[0,1]^J \rightarrow [0,1]\) such that
\[
F_{\mathbf{X}}(x_1,\ldots,x_J)
=
C\!\left(F_{X_1}(x_1),\ldots,F_{X_J}(x_J)\right).
\]
\end{theorem}
This theorem implies that the dependence structure among RVs can be fully captured independently of their marginal distributions. To model spatial correlation among FA ports, we employ the Clayton copula, which captures lower-tail dependence a critical feature in wireless systems dominated by deep fading events~\cite{ghadi2023fluid}. The choice of the Clayton copula is motivated by its ability to capture lower-tail dependence, which corresponds to the joint occurrence of deep fading events across closely spaced FA ports. This property is particularly relevant in OTA-FL systems, where the aggregation performance is dominated by the weakest effective channel conditions. In contrast, copulas such as Gumbel primarily model upper-tail dependence, while the Frank copula does not exhibit explicit tail dependence \cite{silva2008copula}. Therefore, the Clayton copula provides a suitable and tractable model for capturing the dependence structure most critical to the considered system.

For \(J\) dimensions, the Clayton copula is defined as
\begin{equation}
C(u_1,\ldots,u_J)
=
\left(
\sum_{j=1}^{J} u_j^{-\beta} - J + 1
\right)^{-1/\beta},
\end{equation}
where \(\beta \in [-1,\infty)\) is the dependence parameter, with \(\beta = 0\) corresponding to independence.

Under Rayleigh fading, the channel power gain at each FA port, $\{|h_{k,n}|^{2}\}_{n=1}^{N}$, follows an exponential distribution with unit mean, i.e., $\{|h_{k,n}|^{2}\}_{n=1}^{N} \sim \mathrm{Exp}(1)$~\cite{ghadi2023fluid,besser2020copula}. The corresponding CDF and PDF are
\begin{equation}
\label{exp}
F_{|h_{k,n}|^{2}}(x) = 1 - e^{-x}, \quad 
f_{|h_{k,n}|^{2}}(x) = e^{-x}, \quad x \ge 0.
\end{equation}
By leveraging the Clayton copula to model dependence among FA ports, the distribution of the maximum channel power gain is obtained in the
following lemma.

\begin{lemma}
\label{lemma:clayton_cdf}
Assuming \(\{|h_{k,n}|^{2}\}_{n=1}^{N}\) follow correlated exponential distributions modeled by a Clayton copula with parameter \(\beta > -1\), the CDF of the effective FA channel gain
\(
|h_{k}|^{2} = \max_n |h_{k,n}|^{2}
\)
is given by
\begin{equation}
\label{cdf_sk}
F_{|h_{k}|^{2}}(x)
=
\left(
\sum_{i=1}^{N} (1 - e^{-x})^{-\beta}
- N + 1
\right)^{-1/\beta}, \quad x \ge 0.
\end{equation}
\end{lemma}

\begin{proof}
The CDF of $|h_{k}|^{2}$ can be written as
\begin{equation}
\label{proof_step1}
\begin{aligned}
F_{|h_{k}|^{2}}(x)
&= \Pr\!\left( |h_{k}|^{2} < x \right) \\
&= \Pr\!\left( \max \{ |h_{k,1}|^{2}, |h_{k,2}|^{2}, \ldots, |h_{k,N}|^{2} \} < x \right) \\
&= \Pr\!\left( |h_{k,1}|^{2} < x, \ldots, |h_{k,N}|^{2} < x \right) \\
&= F_{|h_{k,1}|^{2}, \ldots, |h_{k,N}|^{2}}(x, \ldots, x).
\end{aligned}
\end{equation}
Invoking Sklar’s theorem, the joint CDF can be expressed in terms of the marginal CDFs and a copula function as
\begin{equation}
F_{|h_{k,1}|^{2},\dots,|h_{k,N}|^{2}}(x,\dots,x) 
= \Big( \sum_{n=1}^{N} F_{|h_{k,n}|^{2}}(x)^{-\beta} - N + 1 \Big)^{-1/\beta}.
\end{equation}
Substituting the exponential CDF from \eqref{exp} completes the proof.
\end{proof}

By substituting the CDF of the channel gain from \eqref{cdf_sk} into \eqref{disMSE}, the CDF of the normalized MSE is obtained as
\begin{equation}
\small
\label{gamma2}
\begin{aligned}
&F_{\gamma}(\tau) =\\& 1 - \sum_{i=0}^{|\mathcal{S}|-1} \binom{K}{i} 
\Bigg[ 1 - \Bigg( \sum_{j=1}^{N}  (1 - e^{-(p_{\max} \tau)^{-1}})^{-\beta}  -N + 1 \Bigg)^{-\frac{1}{\beta}} \Bigg]^{i} \\
&\quad \times \Bigg[ \Bigg( \sum_{j=1}^{N} (1 - e^{-(p_{\max} \tau)^{-1}})^{-\beta}  - N + 1 \Bigg)^{-\frac{1}{\beta}} \Bigg]^{K-i},
\end{aligned}
\end{equation}
Similarly, substituting the channel gain CDF into \eqref{disuser} yields the PMF of the number of participating devices:
\begin{equation}
\label{disuser2}
\begin{aligned}
&\Pr\!\left(\Gamma = |\mathcal{S}|\right)=\\&
\binom{K}{|\mathcal{S}|}
   \big[1 - \left( \sum_{n=1}^{N}
 (1 - e^{\sigma^2(p_{max} \tau)^{-1}})^{-\beta}  - N + 1
\right)^{-\frac{1}{\beta}}\big]^{|\mathcal{S}|} \\& \times 
   \big[\left( \sum_{n=1}^{N}
(1 - e^{\sigma^2(p_{max} \tau)^{-1}})^{-\beta} -N + 1
\right)^{-\frac{1}{\beta}}\big]^{K-|\mathcal{S}|}. 
   \end{aligned}
\end{equation}

Finally, the dependence parameter \(\beta\) provides a quantitative measure of spatial correlation among FA ports. When \(\beta \to 0\), FA ports become statistically independent, yielding maximal selection diversity and optimal learning performance. As \(\beta\) increases, lower-tail dependence intensifies, increasing the likelihood of simultaneous deep fades across FA ports, thereby degrading both aggregation accuracy and user participation. In the asymptotic regime \(\beta \to \infty\), FA ports exhibit perfect positive dependence, and FA-based port selection offers no diversity gain over conventional FPA. This regime defines the fundamental performance limit of FA-assisted OTA-FL under severe spatial correlation.

These results underscore the critical importance of accurately modeling spatial dependence when designing and analyzing FA-assisted OTA-FL systems.

\section{Numerical Results}

In this section, we evaluate the performance of the proposed FA-assisted OTA-FL framework through extensive MC simulations and learning experiments. The impact of spatial correlation among FA ports is explicitly investigated to validate the analytical expressions derived for the CDF of the normalized aggregation MSE, $F_{\gamma}(\tau)$, and the PMF of the number of participating devices, $\Pr(\Gamma = |\mathcal{S}|)$. Unless otherwise stated, all numerical results are obtained by averaging over $10^{4}$ independent MC realizations.

We consider an uplink OTA-FL system with $K=20$ single-antenna UDs, each equipped with an FA architecture comprising $N=10$ selectable antenna ports. The maximum transmit power is set to $p_{\max}=10$~dBm. The noise variance is set to $\sigma^2 = -90$ dBm, and the large-scale fading power is normalized as $\delta^2 = 1$. Spatial correlation among FA ports is controlled by the dependence parameter $\beta$, where smaller values correspond to weaker correlation.

To benchmark the proposed framework, the following configurations are considered:
\begin{itemize}
\item \textbf{Independent FA}: This case corresponds to $\beta \rightarrow 0$, where FA ports experience statistically independent fading, yielding the maximum selection diversity.
\item \textbf{Correlated FA}: This scenario 
considers statistically correlated FA ports with correlation parameters \(\beta = 1\) and \(2\), enabling an investigation of their impact on system performance.

\item \textbf{FPA}: This benchmark corresponds to the limit $\beta \rightarrow \infty$, where all antenna ports experience identical fading and no spatial diversity gain is available.
\end{itemize}

Fig.~\ref{fig:2a} depicts the CDF of the normalized MSE, $F_{\gamma}(\tau)$, for $|\mathcal{S}|=15$ participating devices under different spatial correlation conditions. Analytical curves derived in \eqref{gamma2} are plotted alongside MC simulation results, showing an excellent match across all scenarios and thereby validating the accuracy of the analytical framework.
As expected, $F_{\gamma}(\tau)$ is a monotonically increasing function of the aggregation MSE threshold $\tau$, indicating a higher probability of meeting the target learning accuracy as the threshold is relaxed. It is observed that the independent FA case consistently achieves the steepest CDF curve, reflecting the highest probability of attaining a given MSE target due to maximal spatial diversity. As spatial correlation increases, i.e., for $\beta = 1$ and $\beta = 2$, the CDF curves exhibit a rightward shift, indicating a gradual degradation in learning performance. Nevertheless, even under correlated fading, FA-assisted OTA-FL substantially outperforms the conventional FPA benchmark, which suffers from the absence of antenna selection gain. This degradation is due to the reduced effectiveness of spatial diversity under correlated fading, where channel gains across FA ports become increasingly dependent, thereby limiting the benefit of port selection.

\begin{figure}[htbp]
    \centering
    \includegraphics[width=8cm, height=4.5cm]{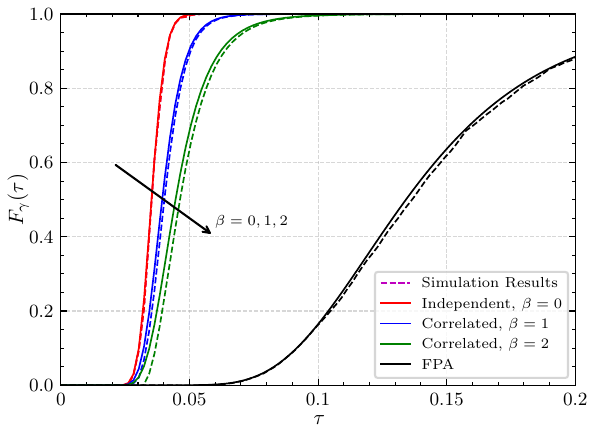}
    \caption{The CDF of the normalized MSE for $|\mathcal{S}|=15$.}
    \label{fig:2a}
\end{figure}

Fig.~\ref{fig:2} presents the PMF of the number of participating devices for a target MSE threshold $\tau=0.05$. The PMF is evaluated over the full range of user participation levels, from no active device ($|\mathcal{S}|=0$) to full participation ($|\mathcal{S}|=20$). The close agreement between analytical and simulated results confirms the correctness of the PMF expression in \eqref{disuser2}. FA-assisted configurations exhibit a significantly higher probability of activating a larger number of devices compared to FPA, demonstrating the critical role of FA-enabled spatial diversity in enhancing user participation.
As spatial correlation increases, i.e., from $\beta = 0$ to $\beta = 2$, the PMF gradually shifts toward smaller values of $|\mathcal{S}|$, reflecting reduced diversity. Nonetheless, even under correlated FA conditions, the achievable participation level remains markedly superior to that of the FPA system. This trend is further highlighted by the mode of the PMF, which is maximized under independent FA and progressively decreases with increasing correlation.
This trend is attributed to the loss of channel diversity under stronger spatial correlation. When FA ports are highly correlated, their channel realizations become similar, reducing the likelihood that port selection can significantly improve the effective channel gain. As a result, fewer users satisfy the participation condition.

\begin{figure}[htbp]
    \centering
    \includegraphics[width=8cm, height=4.5cm]{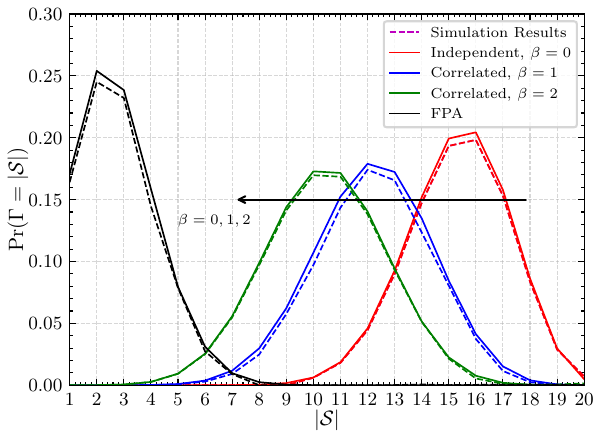}
    \caption{The PMF of the number of participating devices for $\tau = 0.05$.}  
    \label{fig:2}
\end{figure}

Fig.~\ref{fig:3} illustrates the CDF of the normalized MSE for different numbers of FA ports $N$. Increasing $N$ consistently improves $F_{\gamma}(\tau)$ across all configurations, thereby enhancing the probability of satisfying the aggregation MSE constraint. 
This improvement stems from the increased spatial selection diversity afforded by a larger FA aperture. However, the performance gain exhibits diminishing returns as $N$ grows, indicating a saturation effect that limits the incremental diversity obtained from additional ports. This saturation arises because the maximum channel gain approaches its statistical upper limit as $N$ increases, reducing the likelihood of further significant improvement. The independent FA configuration consistently achieves the best performance, while correlated FA scenarios experience a slower improvement rate due to reduced effective diversity. Importantly, all FA-assisted configurations significantly outperform the FPA benchmark over the entire range of $N$.

\begin{figure}[htbp]
    \centering
    \includegraphics[width=8cm, height=4.5cm]{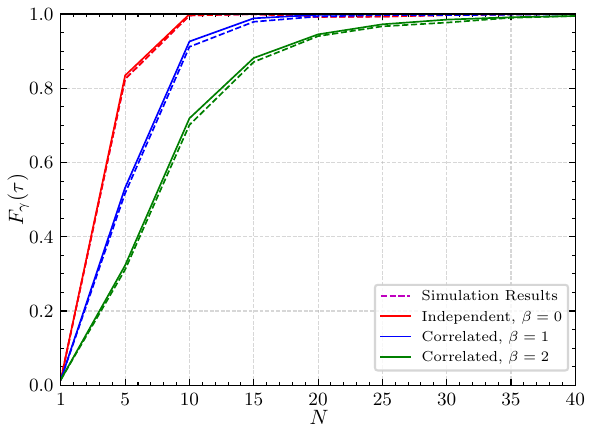}
    \caption{The CDF of the normalized MSE for $|\mathcal{S}|=15$ and $\tau = 0.05$.}  
    \label{fig:3}
\end{figure}

Fig.~\ref{fig:4} reports the probability of achieving full device participation, $\Pr(\Gamma = 20)$, as a function of $N$ for $\tau=0.05$.
It is observed that increasing the number of FA ports $N$ monotonically enhances the probability of full participation, $\Pr(\Gamma = 20)$, across all considered scenarios. 
This improvement is attributed to the increased spatial selection diversity offered by a larger FA aperture. Among the considered cases, the independent FA configuration ($\beta \rightarrow 0$) consistently achieves the highest probability, and approaches unity at a relatively small value of $N$. This indicates that, in the absence of spatial correlation, full-participation OTA-FL can be reliably supported with a limited number of FA ports.
In contrast, scenarios with spatially correlated FA ports ($\beta = 1$ and $\beta = 2$) exhibit a slower increase in $\Pr(\Gamma = 20)$ with respect to $N$, requiring a larger number of ports to approach comparable performance. This behavior reflects the reduction in effective diversity caused by spatial correlation. Nevertheless, even under strong correlation, FA-assisted systems significantly outperform the conventional FPA benchmark, which exhibits a markedly lower probability of full participation over the entire range of $N$. It is also worth noting that the performance improvement becomes less pronounced for large values of $N$, as the benefit of additional ports diminishes once sufficient spatial diversity has been achieved.

\begin{figure}[htbp]
    \centering
    \includegraphics[width=8cm, height=4.5cm]{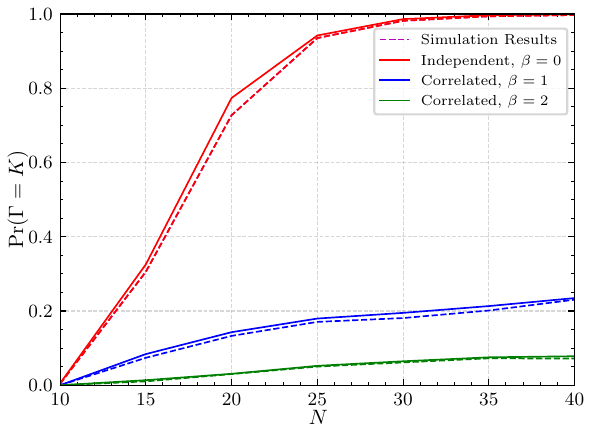}
    \caption{The PMF of full-participation OTA-FL for $\tau = 0.05$ and $|\mathcal{S}| = 20$.}  
    \label{fig:4}
\end{figure}

\subsection{Learning Performance With Real-World Data}

To further investigate the impact of spatial correlation in FAS on OTA-FL performance, we conduct extensive experiments on collaborative image classification using the MNIST dataset. The experimental setup consists of 100 clients, each assigned a distinct local dataset, and the training process is carried out over 100 communication rounds.
The MNIST dataset is split into 90\% training and 10\% testing sets. The training data are further partitioned into 100 non-overlapping subsets and uniformly assigned to the clients, resulting in an IID data distribution. For local training, the Adam optimizer is employed with a learning rate of $\gamma = 0.01$ and a mini-batch size of 32. Each client performs one local gradient update per communication round.

Each client employs a feedforward neural network composed of an input layer with 200 neurons, a single hidden layer with 200 neurons, and an output layer whose dimension matches the number of classes. ReLU activation functions are used in the input and hidden layers, while a softmax activation is applied at the output layer.

Local model updates are aggregated at the server using the proposed OTA-FL framework, which exploits the superposition property of the wireless multiple-access channel. As a benchmark, an ideal FL scenario assuming full device participation and noise-free aggregation is also considered. This benchmark serves as an upper performance bound and isolates the impact of communication impairments and user selection.

The aggregation MSE derived in the previous sections provides a direct link between communication reliability and learning performance in OTA-FL. Specifically, it quantifies the distortion in the aggregated global update caused by imperfect wireless superposition. A lower MSE preserves the accuracy of the aggregated update and enables faster convergence with improved accuracy and loss, whereas a higher MSE introduces noise into the update process and degrades convergence. Therefore, the analytical MSE results directly explain the trends observed in the following learning curves.

Fig.~\ref{fig:5}(a) illustrates the test accuracy versus the communication round index for different antenna configurations and spatial correlation levels. It is observed that FA-assisted OTA-FL consistently outperforms the conventional FPA system. This improvement is primarily attributed to the increased number of selected users per round and the reduced aggregation error enabled by the spatial movement gain of the FAS. Consequently, FA-assisted OTA-FL achieves higher steady-state accuracy, reduced performance fluctuations, and faster convergence.

Fig.~\ref{fig:5}(b) depicts the corresponding training loss evolution. FA-assisted OTA-FL exhibits faster loss decay and lower final loss compared to FPA. Increasing spatial correlation degrades convergence speed and steady-state performance due to reduced selection diversity and higher aggregation distortion. Nevertheless, even under correlated fading, FA-assisted OTA-FL maintains a clear advantage over the FPA benchmark. The gap to the ideal FL curve highlights the impact of wireless aggregation impairments, while confirming the effectiveness of FAS in enhancing OTA-FL performance with real-world data.

\begin{figure}[htbp]
    \centering
    \subfloat[]{
        \includegraphics[width=8cm, height=4.5cm]{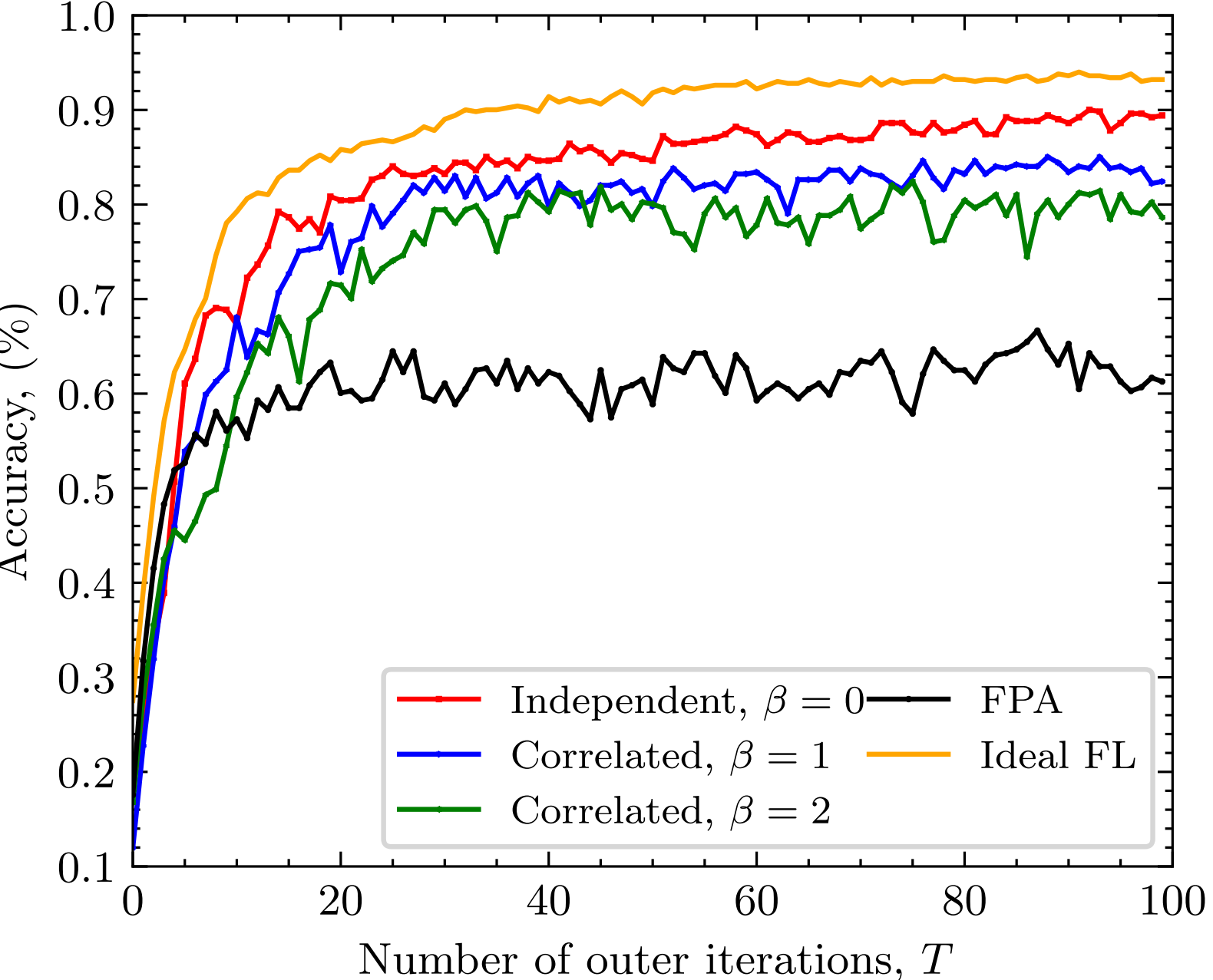}
        \label{fig:5a}
    }
    \hfill
    \subfloat[]{
        \includegraphics[width=8cm, height=4.5cm]{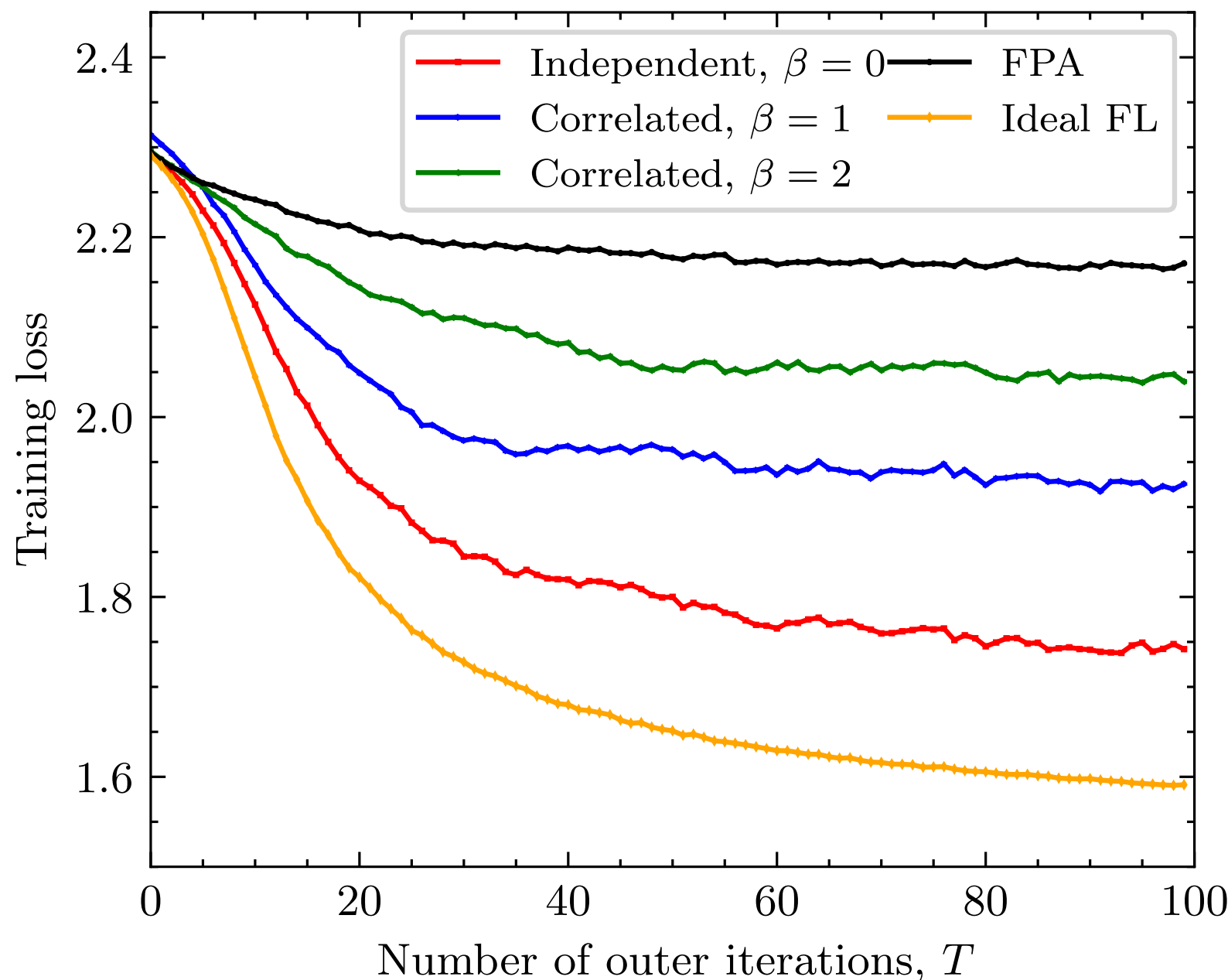}
        \label{fig:5b}
    }
    \caption{Learning performance of FA-assisted OTA-FL on the MNIST dataset with IID data distribution: (a) test accuracy and (b) training loss versus communication rounds.}
    \label{fig:5}
\end{figure}

\section{Conclusion}

This paper investigates the integration of FAS into OTA-FL, analyzing its performance gains and the impact of spatial correlation across FA ports. We establish an analytical framework for FA-assisted OTA-FL and derive tractable expressions for key performance metrics, including the aggregation MSE and the expected number of participating devices per communication round. A copula-based approach is employed to model spatially correlated Rayleigh fading across FA ports, characterizing the effective channel gain and providing new insight into the role of spatial correlation. Numerical simulations and learning experiments on a real-world dataset validate the analysis, demonstrating that FA integration significantly enhances OTA-FL performance, investigates the impact of port correlation, and highlights the effects of FA mobility on global model accuracy.

Despite these gains, practical considerations such as CSI acquisition overhead and unmodeled mutual coupling in dense FA configurations may affect real-world performance. Addressing these aspects through efficient channel acquisition and enhanced channel modeling remains an important direction for future work.

	% argument is your BibTeX string definitions and bibliography database(s)
	\bibliographystyle{IEEEtran}
	\bibliography{Main_Document}

@article{zhao2025fluid,
  title={Fluid Antenna Enabled Over-the-Air Federated Learning: Joint Optimization of Positioning, Beamforming, and User Selection},
  author={Zhao, Yang and Xu, Minrui and Wang, Ping and Niyato, Dusit},
  journal={arXiv preprint arXiv:2503.00011},
  year={2025},
  month={Feb.}
}

@article{besser2020copula,
  title={Copula-based bounds for multi-user communications--Part {II}: Outage Performance},
  author={Besser, Karl-Ludwig and Jorswieck, Eduard A},
  journal={IEEE Communications Letters},
  volume={25},
  month={Jan.},
  number={1},
  pages={8--12},
  year={2021}
}

@book{nelsen2006introduction,
	title={An introduction to copulas},
	author={Nelsen, Roger B},
	year={2006},
	publisher={Springer}
}

@article{ghadi2023fluid,
	title={Fluid antenna-assisted dirty multiple access channels over composite fading},
	author={Ghadi, Farshad Rostami and Wong, Kai-Kit and L{\'o}pez-Mart{\'\i}nez, F Javier and Chae, Chan-Byoung and Tong, Kin-Fai and Zhang, Yangyang},
	journal={IEEE Communications Letters},
	volume={28},
	number={2},
	pages={382--386},
	year={2023},
    month={Dec.}
}

@article{chen2023joint,
	title={Joint client selection and receive beamforming for over-the-air federated learning with energy harvesting},
	author={Chen, Caijuan and Chiang, Yi-Han and Lin, Hai and Lui, John C.S. and Ji, Yusheng},
	journal={IEEE Open Journal of the Communications Society},
	volume={4},
	pages={1127--1140},
	year={2023},
	month={May.},
	publisher={IEEE}
}

@ARTICLE{10538293,
  author={Wang, Zhibin and Zhao, Yapeng and Zhou, Yong and Shi, Yuanming and Jiang, Chunxiao and Letaief, Khaled B.},
  journal={IEEE Internet of Things Journal}, 
  title={Over-the-Air Computation for \text{6G}: Foundations, Technologies, and Applications}, 
month={May.},
  year={2024},
  volume={11},
  number={14},
  pages={24634-24658}}

@ARTICLE{10388035,
  author={Zhu, Jingyang and Shi, Yuanming and Zhou, Yong and Jiang, Chunxiao and Chen, Wei and Letaief, Khaled B.},
  journal={IEEE Internet of Things Journal}, 
  title={Over-the-Air Federated Learning and Optimization}, 
month={Jan.},
  year={2024},
  volume={11},
  number={10},
  pages={16996-17020}}

@INPROCEEDINGS{saeidp,
  author={Pakravan, Saeid and Ahmadzadeh, Mohsen and Zeng, Ming and Abed Hodtani, Ghosheh and Chouinard, J.-Y and Rusch, L. A},
 title={Robust resource allocation for over-the-air computation networks with fluid antenna array},
	author={Pakravan, Saeid and others},
	booktitle={Proc. IEEE Globecom Workshops (GC Wkshps)},
	month={},
	pages={},
	year={},
note={{C}ape {T}own, {S}outh {A}frica, pp. 1--6, Sep. 2024.}
}

@ARTICLE{9829367,
  author={Pham, Quoc-Viet and Zeng, Ming and Huynh-The, Thien and Han, Zhu and Hwang, Won-Joo},
  journal={IEEE Network}, 
  title={Aerial Access Networks for Federated Learning: Applications and Challenges}, 
month={Jul.},
  year={2022},
  volume={36},
  number={3},
  pages={159-166}}

@INPROCEEDINGS{saeidpwcnc,
  author={Ahmadzadeh, Mohsen and Pakravan, Saeid and Abed Hodtani, Ghosheh and Zeng, Ming and Chouinard, Jean-Yves and Rusch, L. A.},
  title={Enhanced Over-The-Air Federated Learning Using \text{AI}-Based Fluid Antenna System},
	booktitle={Proc. IEEE Wireless Communications and Networking Conference (WCNC)},
	month={},
	pages={},
	year={},
note={{M}ilan, {I}taly, pp. 1--6, Mar. 2025.}
}

@article{ahmadzadeh2025ai,
  title={{AI}-based fluid antenna design for client selection in over-the-air federated learning},
  author={Ahmadzadeh, Mohsen and Pakravan, Saeid and Hodtani, Ghosheh Abed and Zeng, Ming and Ye, Qiang John and Chouinard, Jean-Yves and Rusch, Leslie A},
  journal={IEEE Internet of Things Journal},
year={2025},
month={Oct.},
  volume={12},
  number={20},
  pages={42549-42558}
}

@inproceedings{skouroumounis2024swipt,
  title={{SWIPT} in {FA}-enabled cellular networks: A stochastic geometry copula-based approach},
  author={Skouroumounis, Christodoulos and Krikidis, Ioannis},
  	booktitle={Proc. IEEE International Conference on Communications (ICC)},
	month={},
	pages={},
	year={},
note={{D}enver, {USA}, pp. 2360-2365, Aug. 2024.}
}

@article{ahmadzadeh2025advancing,
  title={Advancing Over-the-Air Federated Learning through Deep Reinforcement Learning in {UAV}-Assisted Networks with Movable Antennas},
  author={Ahmadzadeh, Mohsen and Pakravan, Saeid and Hodtani, Ghosheh Abed},
  journal={Computer and Knowledge Engineering},
year={2025},
month={May.},
  volume={8},
  number={2},
  pages={1-10}
}

@article{ghadi2024physical,
 author={Rostami Ghadi, Farshad and Wong, Kai-Kit and Javier López-Martínez, F. and Kiat New, Wee and Xu, Hao and Chae, Chan-Byoung},
  journal={IEEE Transactions on Wireless Communications}, 
  title={Physical Layer Security Over Fluid Antenna Systems: Secrecy Performance Analysis}, 
  year={2024},
  month={Sep.},
  volume={23},
  number={12},
  pages={18201-18213}}

@INPROCEEDINGS{10881858,
  author={Ahmadzadeh, Mohsen and Pakravan, Saeid and Hodtani, Ghosheh Abed and Zeng, Ming and Chouinard, Jean-Yves},
  booktitle={Proc. of IEEE Middle East Conference on Communications and Networking (\text{MECOM})}, 
  title={Deep Reinforcement Learning for Robust \text{RIS}-Aided Over-the-Air Federated Learning in Cognitive Radio}, 
booktitle={Proc. IEEE Middle East Conference on Communications and Networking (MECOM)},
	month={},
	pages={},
	year={},
note={{A}bu {D}habi, {U}nited {A}rab {E}mirates, pp. 368-373, Nov. 2024.}
}

@INPROCEEDINGS{7777777,
  author={X. Wang and Y. Chen and Q. Ye and O. A. Dobre},
  title={Connectivity Enrichment for Decentralized Federated Learning Networks with Teleportation}, 
	booktitle={Proc. IEEE International Conference on Computer Communications (\text{INFOCOM WKSHPS})},
	month={},
	pages={},
	year={},
note={ {L}ondon, {U}nited {K}ingdom, pp. 1--6, May. 2025.}
}

@ARTICLE{18952884,
  author={Yang, Kai and Jiang, Tao and Shi, Yuanming and Ding, Zhi},
  journal={IEEE Transactions on Wireless Communications}, 
  title={Federated Learning via Over-the-Air Computation}, 
  year={2020},
month={Jan.},
  volume={19},
  number={3},
  pages={2022-2035}}

@ARTICLE{10092857,
  author={Şahin, Alphan and Yang, Rui},
  journal={IEEE Communications Surveys \& Tutorials}, 
  title={A Survey on Over-the-Air Computation}, 
  month={Apr.},
  year={2023},
  volume={25},
  number={3},
  pages={1877-1908}}

@ARTICLE{9141214,
  author={Niknam, Solmaz and Dhillon, Harpreet S. and Reed, Jeffrey H.},
  journal={IEEE Communications Magazine}, 
  title={Federated Learning for Wireless Communications: Motivation, Opportunities, and Challenges}, 
  year={2020},
  month={Jun.},
  volume={58},
  number={6},
  pages={46-51}}

@ARTICLE{9264742,
  author={Yang, Zhaohui and Chen, Mingzhe and Saad, Walid and Hong, Choong Seon and Shikh-Bahaei, Mohammad},
  journal={IEEE Transactions on Wireless Communications}, 
  title={Energy Efficient Federated Learning Over Wireless Communication Networks}, 
  year={2021},
  month={Nov.},
  volume={20},
  number={3},
  pages={1935-1949}}

@ARTICLE{10988827,
  author={Mohajer Hamidi, Shayan and Bereyhi, Ali and Asaad, Saba and Vincent Poor, H.},
  journal={IEEE Communications Letters}, 
  title={Over-the-Air Fair Federated Learning via Multi-Objective Optimization}, 
  year={2025},
  month={May.},
  volume={29},
  number={7},
  pages={1549-1553}}

@ARTICLE{10932699,
  author={Yan, Na and Wang, Kezhi and Zhi, Kangda and Pan, Cunhua and Chai, Kok Keong and Poor, H. Vincent},
  journal={IEEE Transactions on Wireless Communications}, 
  title={Secure and Private Over-the-Air Federated Learning: Biased and Unbiased Aggregation Design}, 
  year={2025},
  month={Mar.},
  volume={24},
  number={7},
  pages={5901-5916}}

@ARTICLE{9738803,
  author={Nabil, Yasser and ElSawy, Hesham and Al-Dharrab, Suhail and Mostafa, Hassan and Attia, Hussein},
  journal={IEEE Internet of Things Journal}, 
  title={Data Aggregation in Regular Large-Scale {IoT} Networks: Granularity, Reliability, and Delay Tradeoffs}, 
  year={2022},
   month={Mar.},
  volume={9},
  number={18},
  pages={17767-17784}}

@ARTICLE{10261444,
  author={Sifaou, Houssem and Li, Geoffrey Ye},
  journal={IEEE Transactions on Wireless Communications}, 
  title={Over-The-Air Federated Learning Over Scalable Cell-Free Massive {MIMO}}, 
  year={2024},
   month={Sep.},
  volume={23},
  number={5},
  pages={4214-4227}}

@ARTICLE{10444710,
  author={Asaad, Saba and Tabassum, Hina and Ouyang, Chongjun and Wang, Ping},
  journal={IEEE Transactions on Wireless Communications}, 
  title={Joint Antenna Selection and Beamforming for Massive {MIMO}-Enabled Over-the-Air Federated Learning}, 
  year={2024},
  month={Feb.},
  volume={23},
  number={8},
  pages={8603-8618}}

@ARTICLE{10381609,
  author={Liu, Hang and Yan, Jia and Zhang, Ying-Jun Angela},
  journal={IEEE Transactions on Wireless Communications}, 
  title={Differentially Private Over-the-Air Federated Learning Over {MIMO} Fading Channels}, 
  year={2024},
  month={Jan.},
  volume={23},
  number={8},
  pages={8232-8247}}

@ARTICLE{9829190,
  author={Zheng, Jingheng and Tian, Hui and Ni, Wanli and Ni, Wei and Zhang, Ping},
  journal={IEEE Transactions on Wireless Communications}, 
  title={Balancing Accuracy and Integrity for Reconfigurable Intelligent Surface-Aided Over-the-Air Federated Learning}, 
  year={2022},
  month={Jul.},
  volume={21},
  number={12},
  pages={10964-10980}}

@ARTICLE{10912454,
  author={Zhang, Xinran and Tian, Hui and Ni, Wanli and Yang, Zhaohui},
  journal={IEEE Internet of Things Journal}, 
  title={Joint Beamforming Design for Multifunctional {RIS}-Aided Over-the-Air Federated Learning}, 
  year={2025},
  month={Jun.},
  volume={12},
  number={12},
  pages={21720-21739}}

@ARTICLE{10529197,
  author={Kim, Minsik and Park, Daeyoung},
  journal={IEEE Wireless Communications Letters}, 
  title={Reconfigurable Intelligent Surfaces-Aided Federated Learning in Over-the-Air Computation}, 
  year={2024},
   month={May.},
  volume={13},
  number={7},
  pages={1983-1987}}

@ARTICLE{11175437,
  author={Hong, Hanjiang and Wong, Kai-Kit and Chae, Chan-Byoung and Xu, Hao and Guo, Xinghao and Ghadi, Farshad Rostami and Chen, Yu and Xu, Yin and Liu, Baiyang and Tong, Kin-Fai and Zhang, Yangyang},
  journal={IEEE Transactions on Network Science and Engineering}, 
  title={A Contemporary Survey on Fluid Antenna Systems: Fundamentals and Networking Perspectives}, 
  year={2025},
  month={Sep.},
  volume={13},
  number={},
  pages={2305-2328}}

@ARTICLE{9715064,
  author={Chai, Zhi and Wong, Kai-Kit and Tong, Kin-Fai and Chen, Yu and Zhang, Yangyang},
  journal={IEEE Communications Letters}, 
  title={Port Selection for Fluid Antenna Systems}, 
  year={2022},
  month={Feb.},
  volume={26},
  number={5},
  pages={1180-1184}}

@ARTICLE{10146274,
  author={Wong, Kai-Kit and New, Wee Kiat and Hao, Xu and Tong, Kin-Fai and Chae, Chan-Byoung},
  journal={IEEE Communications Letters}, 
  title={Fluid Antenna System—Part {I}: Preliminaries}, 
  year={2023},
  month={Jul.},
  volume={27},
  number={8},
  pages={1919-1923}}

@ARTICLE{10753482,
  author={New, Wee Kiat and Wong, Kai-Kit and Xu, Hao and Wang, Chao and Ghadi, Farshad Rostami and Zhang, Jichen and Rao, Junhui and Murch, Ross and Ramírez-Espinosa, Pablo and Morales-Jimenez, David and Chae, Chan-Byoung and Tong, Kin-Fai},
  journal={IEEE Communications Surveys \& Tutorials}, 
  title={A Tutorial on Fluid Antenna System for {6G} Networks: Encompassing Communication Theory, Optimization Methods and Hardware Designs}, 
  year={2025},
  month={Nov.},
  volume={27},
  number={4},
  pages={2325-2377}}

@ARTICLE{9998496,
  author={Van Chien, Trinh and Lagunas, Eva and Hoang, Tiep M. and Chatzinotas, Symeon and Ottersten, Björn and Hanzo, Lajos},
  journal={IEEE Transactions on Communications}, 
  title={Space-Terrestrial Cooperation Over Spatially Correlated Channels Relying on Imperfect Channel Estimates: Uplink Performance Analysis and Optimization}, 
  year={2023},
  month={Dec.},
  volume={71},
  number={2},
  pages={773-791}}

@ARTICLE{10623405,
  author={Ramírez-Espinosa, Pablo and Morales-Jimenez, David and Wong, Kai-Kit},
  journal={IEEE Transactions on Wireless Communications}, 
  title={A New Spatial Block-Correlation Model for Fluid Antenna Systems}, 
  year={2024},
  month={Aug.},
  volume={23},
  number={11},
  pages={15829-15843}}

@ARTICLE{10716282,
  author={Lai, Xiazhi and Yao, Junteng and Zhi, Kangda and Wu, Tuo and Morales-Jimenez, David and Wong, Kai-Kit},
  journal={IEEE Transactions on Vehicular Technology}, 
  title={{FAS-RIS}: A Block-Correlation Model Analysis}, 
  year={2025},
  month={Oct.},
  volume={74},
  number={2},
  pages={3412-3417}}

@ARTICLE{11106811,
  author={Pakravan, Saeid and Ahmadzadeh, Mohsen and Zeng, Ming and Yang, Zhaohui and Hodtani, Ghosheh Abed and Chouinard, Jean-Yves and Pham, Quoc-Viet},
  journal={IEEE Transactions on Vehicular Technology}, 
  title={Fluid Antenna-Assisted Uplink {NOMA} Networks Under Imperfect {SIC}}, 
  year={2025},
   month={Aug.},
  volume={75},
  number={1},
  pages={1689-1694}}

@article{pakravan2026fluid1,
  title={Fluid Antenna Systems under Channel Uncertainty and Hardware Impairments: Trends, Challenges, and Future Research Directions},
  author={Pakravan, Saeid and Ahmadzadeh, Mohsen and Zeng, Ming and Ajib, Wessam and Wang, Ji and Li, Xingwang},
  journal={arXiv preprint arXiv:2601.22989},
  year={2026},
  month={Jan.}
}

@article{silva2008copula,
  author    = {Silva, Ralph dos Santos and Hedibert Freitas Lopes},
  title     = {Copula, marginal distributions and model selection: A Bayesian note},
  journal   = {Statistics and Computing},
  volume    = {18},
  number    = {3},
  pages     = {313--320},
  year      = {2008},
  month     = {Mar.}
}

@ARTICLE{10258328,
  author={Bodet, Duschia M. and Jornet, Josep M.},
  journal={IEEE Open Journal of the Communications Society}, 
  title={Directional Antennas for Sub-{THz} and {THz} {MIMO} Systems: Bridging the Gap Between Theory and Implementation}, 
  year={2023},
  month     = {Sep.},
  volume={4},
  number={},
  pages={2261-2273}}

@ARTICLE{10906511,
  author={Zhu, Lipeng and Ma, Wenyan and Mei, Weidong and Zeng, Yong and Wu, Qingqing and Ning, Boyu and Xiao, Zhenyu and Shao, Xiaodan and Zhang, Jun and Zhang, Rui},
  journal={IEEE Communications Surveys \& Tutorials}, 
  title={A Tutorial on Movable Antennas for Wireless Networks}, 
  year={2025},
  month     = {Feb.},
  volume={28},
  number={},
  pages={3002-3054}}

@ARTICLE{9849110,
  author={Peng, Jinlin and Tang, Wenfei and Zhang, Hongtao},
  journal={IEEE Wireless Communications Letters}, 
  title={Directional Antennas Modeling and Coverage Analysis of {UAV}-Assisted Networks}, 
  year={2022},
  month     = {Aug.},
  volume={11},
  number={10},
  pages={2175-2179}}

@ARTICLE{9625822,
  author={Elbir, Ahmet M. and Coleri, Sinem},
  journal={IEEE Transactions on Wireless Communications}, 
  title={Federated Learning for Channel Estimation in Conventional and {RIS}-Assisted Massive {MIMO}}, 
  year={2022},
    month     = {Nov.},
  volume={21},
  number={6},
  pages={4255-4268}}

@ARTICLE{9626135,
  author={Ni, Wanli and Liu, Yuanwei and Yang, Zhaohui and Tian, Hui and Shen, Xuemin},
  journal={IEEE Internet of Things Journal}, 
  title={Federated Learning in Multi-{RIS}-Aided Systems}, 
  year={2022},
  month     = {Nov.},
  volume={9},
  number={12},
  pages={9608-9624}}

@ARTICLE{10649032,
  author={Sun, Pengcheng and Liu, Erwu and Ni, Wei and Wang, Rui and Xing, Zhe and Li, Bofeng and Jamalipour, Abbas},
  journal={IEEE Transactions on Communications}, 
  title={Reconfigurable Intelligent Surface-Assisted Wireless Federated Learning With Imperfect Aggregation}, 
  year={2025},
  month     = {Feb.},
  volume={73},
  number={2},
  pages={1058-1071}}

@ARTICLE{11342412,
  author={Ahmadzadeh, Mohsen and Pakravan, Saeid and Hodtani, Ghosheh Abed and Zeng, Ming and Li, Xingwang and Zhang, Ning and Chouinard, Jean-Yves},
  journal={IEEE Transactions on Vehicular Technology}, 
  title={{AI}-Enhanced {RIS}-Aided Cognitive Radio Network: Integrating Communication and Over-the-Air Federated Learning Users}, 
  year={2026},
  month     = {Jan.},
  volume={},
  number={},
  pages={1-14}}

@ARTICLE{10363658,
  author={Pakravan, Saeid and Chouinard, Jean-Yves and Zeng, Ming and Li, Xingwang and Hao, Wanming and Dobre, Octavia A.},
  journal={IEEE Internet of Things Journal}, 
  title={Physical-Layer Security of {RIS}-Assisted Networks Over Correlated Fisher-Snedecor {F} Fading Channels}, 
  year={2024},
  month     = {May.},
  volume={11},
  number={9},
  pages={15152-15165}}

	\vfill
	
\end{document}